\documentclass[11pt,a4paper]{article}

\usepackage[margin=25mm]{geometry}
\usepackage[T1]{fontenc}
\usepackage[utf8]{inputenc}
\usepackage{lmodern}
\usepackage[expansion=false]{microtype}
\usepackage{amsmath,amssymb,amsthm}
\usepackage[notrig]{physics}

\usepackage{flafter}
\usepackage{tikz}
\usetikzlibrary{arrows.meta}
\usepackage{needspace}
\definecolor{figblue}{HTML}{24637A}
\definecolor{figorange}{HTML}{B76A31}
\definecolor{figink}{HTML}{303B43}
\definecolor{figgray}{HTML}{6B747B}
\tikzset{
  paper box/.style={draw=figink,line width=0.6pt,
    minimum height=9mm,align=center,font=\small,fill=white,inner sep=5pt},
  paper arrow/.style={-{Latex[length=1.6mm,width=1.1mm]},draw=figink},
  paper classical/.style={double,double distance=1pt,line width=0.45pt},
  paper label/.style={font=\small,align=center,text=figink},
  paper panel/.style={font=\small\bfseries,anchor=west,text=figink}
}

\usepackage[backend=biber,style=numeric-comp,sorting=none,maxbibnames=99]{biblatex}
\usepackage{hyperref}
\hypersetup{
  hidelinks,
  pdftitle={Certifying Measurement Incompatibility under Bounded Classical Communication},
  pdfauthor={Oxana Shaya}
}
\newtheorem{theorem}{Theorem}[section]
\newtheorem{proposition}[theorem]{Proposition}
\newtheorem{lemma}[theorem]{Lemma}
\newtheorem{corollary}[theorem]{Corollary}
\newcommand{\recalledname}{}
\newtheorem*{recalledstatement}{\recalledname}
\newenvironment{recalled}[3]
 {\renewcommand{\recalledname}{#1~\ref{#2}}\begin{recalledstatement}[#3]}
 {\end{recalledstatement}}
\theoremstyle{definition}
\newtheorem{assumption}{Assumption}
\theoremstyle{plain}

\newcommand{\Acup}{\Lambda^{\cup}}
\newcommand{\id}{I}

\title{Certifying Measurement Incompatibility\\
under Bounded Classical Communication}
\author{Oxana Shaya\\[0.4em]
\small Institute for Information Processing (tnt/L3S)\\
\small Leibniz Universit\"at Hannover, Germany}
\date{}
\begin{document}
\maketitle
\begin{abstract}
Measurement incompatibility enables quantum advantages in communication
and state discrimination. It is also necessary for Bell nonlocality, which is used in device-independent protocols for quantum key distribution and private randomness generation. Classical communication can influence
the receiver's outcome and must therefore be accounted for when certifying
measurement incompatibility. We introduce union maximal leakage, which quantifies the
information needed for a common classical explanation of all observed
probabilities and is computed from them by a linear program. If a message of at
most \(k\) bits precedes the receiver's independent
measurement choice, compatible measurements for each message give
union maximal leakage at most \(k\). This bound holds with any
input-independent nonsignalling shared resource. In quantum theory, a violation
also implies that the initially shared state is entangled. Conversely,
union maximal leakage determines the minimum classical communication needed to reproduce
the observations with compatible measurements and arbitrary nonsignalling
assistance. For a fixed shared quantum state and receiver measurements,
a violation is possible exactly when, for some message value,
the receiver's measurements are Bell-nonlocal against two-outcome sender
measurements on that state.  For binary receiver outcomes, we show that a violation also certifies randomness private from the sender, even when the receiver’s measurement depends on the sender’s message and the sender retains quantum side information. We quantify this
privacy through an explicit conditional entropy bound and establish
composably secure randomness extraction in sequential protocols with
device memory.
\end{abstract}

\section{Introduction}
\label{sec:introduction}

Measurements are compatible if a single parent measurement, followed by
suitable classical processing of its outcome, reproduces the statistics
of each of them for every input state. Incompatibility enables advantages in
communication \cite{PR_Saha2023} and state discrimination
\cite{PR_Carmeli2019,PR_Skrzypczyk2019}, and is necessary for Bell
nonlocality \cite{PR_Wolf2009,PR_Guhne2023}. Applications based on Bell
violations, including device-independent quantum key distribution
\cite{PR_Acin2007} and private randomness generation
\cite{PR_Pironio2010Randomness}, therefore rely on incompatible
measurements.

When the measurement operators are known, compatibility can be determined by testing whether the measurements admit a common parent measurement. For projective measurements, this is equivalent to their projectors commuting pairwise.
Different states and measurements can produce the same observed probabilities.
In a Bell experiment with
independent measurement choices and no communication, compatible measurements
on either side yield a local hidden-variable model. A Bell inequality
violation therefore rules out compatibility
\cite[Sec.~IV.A]{PR_Guhne2023}. If a classical message is allowed,
however, the receiver can learn about the sender's input and use that
information to select a measurement or produce an output. The local
bounds of Bell inequalities need no longer hold, even with compatible
measurements.

Measurement incompatibility has already been tested in entanglement-assisted random-access experiments through the probability of correctly recovering a requested input bit \cite{PR_Xiao2021EARAC,PR_Zhang2026Instruments}. Here we approach the problem through the information
required to explain all the receiver's measurement statistics together.
The question is whether a common classical explanation can reproduce
the observations within the communication allowed in the experiment.
The resulting test includes the standard random-access certificate with
two input bits and one communicated bit: with uniformly chosen inputs,
every violation of its compatibility bound
is also detected by our criterion (Appendix~\ref{app:basic}).

We consider the one-way experiment in Fig.~\ref{fig:scenario}.
The sender and receiver share a resource that is independent of their
inputs and obeys nonsignalling, so it cannot by itself transmit
information between them. On receiving an input \(x\), the sender
sends a classical message \(m\in\mathcal M\). The receiver then
independently chooses a measurement setting \(y\) and reports an outcome
\(b\), with conditional probabilities \(p(b\mid x,y)\).
The receiver's measurement may depend on \(m\), and the message is the
only communication. 

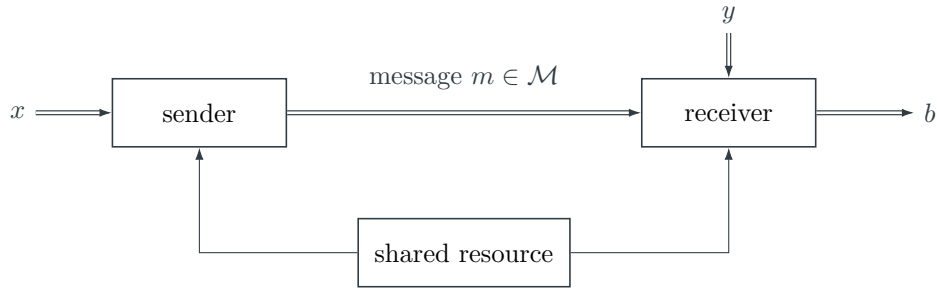
\begin{figure}[ht]
\centering
\begin{tikzpicture}[x=1cm,y=1cm]
 \node[paper box,minimum width=23mm] (sender) at (3,0.7) {sender};
 \node[paper box,minimum width=23mm] (receiver) at (10,0.7) {receiver};
 \node[paper box,minimum width=28mm] (shared) at (6.5,-1.15)
       {shared resource};
 \draw[paper arrow] (shared.west)--(3,-1.15)--(sender.south);
 \draw[paper arrow] (shared.east)--(10,-1.15)--(receiver.south);
 \node[paper label] (input) at (0.6,0.7) {\(x\)};
 \draw[paper classical,paper arrow] (input.east)--(sender.west);
 \draw[paper classical,paper arrow] (sender.east)--(receiver.west)
       node[midway,above=5pt,paper label] {message \(m\in\mathcal M\)};
 \node[paper label] (setting) at (10,2) {\(y\)};
 \draw[paper classical,paper arrow] (setting.south)--(receiver.north);
 \draw[paper classical,paper arrow] (receiver.east)--(12.45,0.7)
       node[right,paper label] {\(b\)};
\end{tikzpicture}
\caption{Communication scenario. The inputs \(x\) and \(y\) are independent
of each other and of the shared resource. The receiver chooses \(y\)
and obtains \(b\) after receiving the message \(m\).
Double lines carry classical information.
The message may be public and is the only communication. For \(k\)
fixed-length bits, \(\abs{\mathcal M}\le2^k\).}
\label{fig:scenario}
\end{figure}

In this setting, can we certify from the observed probabilities that the receiver’s measurements are incompatible?
The parent-measurement description suggests how to answer this question.
If the receiver's measurements are compatible for each message, the
message and the outcome of the corresponding parent measurement provide
a common classical explanation of all observed probabilities. Such an explanation always
exists if it is allowed to encode the full sender input. The question is
how little information about \(x\) it needs to carry, and whether this
exceeds what the permitted communication can supply. We quantify this
information by maximal leakage \cite{PR_Issa2020} and minimize over all
classical simulations reproducing the observed probabilities. Following
the union-information construction of Ref.~\cite{PR_Kolchinsky2022}, we
call the resulting quantity \emph{union maximal leakage}.

Writing \(\log_2\Acup(p)\) for this leakage, we prove that compatible
receiver measurements satisfy
\begin{equation}
 \Acup(p)\le\abs{\mathcal M}\le2^k
 \label{intro:bound}
\end{equation}
for a message of at most \(k\) fixed-length bits. A violation rules out
compatibility for at least one message value. The bound holds in any
nonsignalling operational theory and, in quantum theory, a violation also
implies that the initial shared state is entangled. A linear program
computes \(\Acup(p)\) directly from the observed probabilities.
The same quantity also determines how much classical communication a
compatible explanation needs when the nonsignalling resource may be
chosen freely. For a fixed quantum state and receiver measurements,
the question is different: can any choice of sender operations make
the test detect incompatibility? We characterize this possibility
through Bell nonlocality. We then investigate what a violation says
about privacy when the sender retains quantum information correlated
with the receiver's output.

\Needspace{8\baselineskip}
Section~\ref{sec:test} defines union maximal leakage, derives the
compatibility bound and its converse, and gives the linear program.
Section~\ref{sec:frontier} characterizes when a fixed shared state and
receiver measurements can violate the bound.
Section~\ref{sec:privacy} applies the test to private randomness,
first bounding the receiver's conditional output entropy and then
establishing secure extraction over repeated rounds.
Section~\ref{sec:outlook} concludes with the remaining questions.

\subsection{Relation to prior work}
\label{sec:priorwork}
Information-theoretic incompatibility tests already use the fact that
a parent measurement must be at least as informative as each measurement
obtained from it by classical processing. Zhu's criteria
\cite{PR_Zhu2015} compare Fisher information about quantum-state
parameters with bounds from quantum estimation theory. Their evaluation
uses a specified state-estimation model or, in an equivalent formulation,
the measurement operators. Related criteria use properties of measurement
operators \cite{PR_HJN2022} or evaluate information measures on conditional
quantum states, measurements and instruments \cite{PR_Lee2026}.
Our test addresses experiments in which these quantum descriptions are
unavailable but the classical communication can be bounded. The sender's
input \(x\) is a classical label and need not parameterize a known family
of quantum states. Under our experimental assumptions, nonsignalling and
the message bound constrain the information that compatible receiver
measurements can reveal about \(x\). This gives a certificate evaluated
directly from \(p(b\mid x,y)\), without reconstructing the state or
measurements, even when the receiver uses the message to choose its
measurement.

Our optimization has a more direct
connection to Tavakoli \emph{et al.}~\cite{PR_Tavakoli2022}, who study
prepare-and-measure experiments with a bound on the receiver's optimal
probability of guessing the sender's input. For \(N\)
uniformly distributed inputs, minimizing this probability over classical
models that reproduce \(p\) gives \(\Acup(p)/N\). Their linear program
decides whether such a model exists for a specified guessing-probability
bound \cite[Sec.~3.2]{PR_Tavakoli2022}. Our formulation uses fewer variables by optimizing directly over
joint outcome distributions. To turn this quantity into an incompatibility test, we
use the communication model: nonsignalling and the number of possible
messages imply the information bound, so the guessing probability need
not be assumed or estimated separately.

The interpretation in terms of communication depends on which shared
resources are allowed. With shared randomness, the required number of
classical message values is studied through signalling dimension
\cite{PR_DallArno2017,PR_DoolittleChitambar2021}. For a single receiver
setting, \(\lceil\Acup(p)\rceil\) is the maximum-likelihood lower bound
on this number \cite{PR_DoolittleChitambar2021}. With arbitrary
nonsignalling assistance and compatible receiver measurements,
Theorem~\ref{thm:communication}(ii) makes the bound attainable, including
when there are several settings. Without the compatibility requirement,
the simulation criterion bounds the communication value of each receiver
setting separately \cite[Result~3]{PR_Nogueira2026NS}. Compatibility
requires a common simulation of all settings within the same communication
bound. Here joint outcome distributions exist for every sender input, but the
information they carry about that input is restricted. Entropic
marginal-consistency tests instead constrain the existence of joint
distributions themselves \cite{PR_FritzChaves2013}.

For incompatibility certification in the communicating experiment,
entanglement-assisted random-access tests are the closest comparison
\cite{PR_Xiao2021EARAC,PR_Zhang2026Instruments}. The correlations between
the sender and the first receiver in these protocols fit our quantum
setting with one classical bit. Their full sequential protocols also characterize
measurement instruments, which is a separate task. Our linear program includes the
standard two-bit random-access certificate with uniformly chosen
inputs, while allowing the full observed distribution to determine
the linear certificate. Bell experiments also permit tests of more general
compatibility structures through linear and semidefinite optimization
\cite{PR_Quintino2019Structures,PR_Chen2021Incompatibility}.
When the sender inputs, sender measurements and receiver output
processing may be chosen as in Theorem~\ref{thm:operationalfrontier},
Bell nonlocality against two-outcome sender measurements determines
exactly which fixed quantum resources can violate our test.
Known Bell-nonlocal families \cite{PR_Bene2018} therefore give detection
beyond pairs.

\Needspace{3\baselineskip}
The communication restriction also determines how the test relates to
other experimental approaches. These include tests with trusted quantum
inputs \cite{PR_Guerini2019}, or tests that restrict
the dimension or information capacity of a transmitted quantum system
\cite{PR_Saha2023,PR_Carmeli2020RAC,PR_HeinosaariLeppajarvi2022}.
These assumptions concern the quantum systems supplied or transmitted
in those experiments, whereas our main setting bounds classical
messages. Our extension to
quantum messages connects the two settings. For a fixed measurement,
Frenkel and Weiner show that the statistics of \(d\)-level quantum
states can be simulated with \(d\) classical message values and shared
randomness \cite{PR_FrenkelWeiner2015}. Saha \emph{et al.}~\cite{PR_Saha2023}
show that, without preshared entanglement, compatible measurements on a
transmitted \(d\)-dimensional system with shared randomness yield exactly
the observations admitting such a classical simulation. In this setting,
our condition \(\Acup(p)>d\) is a sufficient test, whereas their result
characterizes all observations obtainable with compatible measurements
under those assumptions.

Privacy against another participant with quantum memory is already
the goal of blind randomness generation \cite{PR_MillerShi2017}, and generalised entropy
accumulation establishes security over repeated rounds
\cite[Sec.~5.1]{PR_Metger2024}. Communication also appears in randomness
protocols where one party sends its measurement outcome to select the
other party's measurement. These protocols certify the parties' joint
outputs against an external adversary \cite{PR_Agresti2020}.
Randomness certification with classical communication is also studied in
entanglement-assisted random-access protocols
\cite[Sec.~V]{PR_Xiao2021EARAC}. Measurement-incompatibility
witnesses have also been used to bound randomness with trusted quantum
inputs \cite{PR_Schloesser2026Randomness}. In our application the sender
both communicates with the receiver and retains quantum information,
while only the receiver's output is to be kept private. Our additional
result is the explicit conditional entropy bound obtained from the
union-leakage violation in this setting. Generalised entropy accumulation
and privacy amplification then convert this bound into composable
security for sequential protocols with device memory.

\section{Union maximal leakage and the incompatibility test}
\label{sec:test}

\subsection{Classical simulations and union maximal leakage}
\label{sec:notation}
All input and outcome alphabets are finite and nonempty. The sender's input is
\(x\in\{1,\ldots,N\}\), the receiver's measurement setting is
\(y\in\{1,\ldots,n\}\), and the reported outcome is
\(b\in\mathcal B_y\). Their conditional probabilities satisfy
\(p(b\mid x,y)\ge0\) and \(\sum_b p(b\mid x,y)=1\).
A classical channel is a conditional distribution normalized over its
output for every input. We use \(q\) for a simulation channel and \(\mathcal N\)
for a physical communication channel. The symbol \(\abs{\mathcal S}\)
denotes the cardinality of a finite set. For sets, \(\prod_{y=1}^n\mathcal B_y\)
denotes the Cartesian product \(\mathcal B_1\times\cdots\times\mathcal B_n\).
Logarithms are to base two. For \(u\in\mathbb R\), we write
\([u]_+:=\max\{u,0\}\) for its positive part.
For classical random variables, uppercase letters denote the variables
and lowercase letters their values, as in \(X=x\).
Party subscripts \(A\) and \(B\) identify the sender and receiver,
respectively.
In quantum formulas, \(\rho\) is a normalized state, \(\sigma\) may be
subnormalized, \(A\succeq B\) means that \(A-B\) is positive
semidefinite, and \(\id\) is the identity. The trace norm is
\(\norm{A}_1=\tr\sqrt{A^\dagger A}\). A positive operator-valued
measure (POVM) consists of positive effects summing to \(\id\).

A common classical simulation encodes the input value \(x\) in a finite
random variable \(Z\), with values \(z\in\mathcal Z\). Given \(z\)
and the chosen measurement setting \(y\), a response channel generates
\(b\):
\begin{equation}
 p(b\mid x,y)=\sum_{z\in\mathcal Z}q(z\mid x)p_B(b\mid y,z).
 \label{eq:simulation}
\end{equation}
The alphabet \(\mathcal Z\) is unrestricted. To quantify what \(Z\)
reveals about the input, we use a uniform reference distribution for \(X\).
Without observing \(Z\), the optimal probability of guessing \(X\)
is \(1/N\). For a given output \(z\), an optimal guess chooses an input
maximizing \(q(z\mid x)\). Summing these maxima and dividing by \(N\)
gives the optimal guessing probability \(P_g(X\mid Z)\). Their sum is
therefore the multiplicative increase in the optimal guessing probability
due to observing \(Z\):
\begin{equation}
 \Lambda(q):=\sum_z\max_x q(z\mid x)
 =N P_g(X\mid Z)=2^{\mathcal L_{\max}(X\to Z)}.
 \label{eq:guessinggain}
\end{equation}
Here \(\mathcal L_{\max}(X\to Z)\) is maximal leakage
\cite{PR_Issa2020}, and \(\Lambda(q)\) is also known as the classical
communication value \cite{PR_Chitambar2023CV}.
The uniform prior relates \(\Lambda(q)\) to the guessing probability
\cite{PR_Tavakoli2022}, without restricting how experimental inputs
are sampled.

Given \(z\), the response channel specifies an outcome distribution
for every receiver setting. Sampling from these distributions gives a
joint outcome \(s=(s_1,\ldots,s_n)\in\Omega:=\prod_{y=1}^n\mathcal B_y\).
The simulation reports the component \(b=s_y\) for the chosen setting \(y\).
The conditional distribution \(q(s\mid x)\) must then have marginals
\(p(b\mid x,y)\).
Different joint distributions can reproduce these same observations while
revealing different amounts of information about the input. To find the
least information required by a common explanation, we minimize
\(\Lambda(q)\) over all such distributions. Writing \(p\) for the full
family of observed probabilities, this gives
\begin{equation}
\begin{aligned}
 \mathcal C(p)&:=\left\{q\in\mathbb R_{\ge0}^{N\times\abs{\Omega}}:
  \sum_{s\in\Omega:\,s_y=b}q(s\mid x)=p(b\mid x,y)
  \quad\forall x,y,b\right\},\\
 \Acup(p)&:=\min_{q\in\mathcal C(p)}\sum_{s\in\Omega}\max_x q(s\mid x).
 \label{eq:operationalcosts}
\end{aligned}
\end{equation}
For each input \(x\), \(q(\cdot\mid x)\) is a probability distribution:
summing the marginal equalities over \(b\) gives
\(\sum_{s\in\Omega}q(s\mid x)=1\).
One admissible joint distribution is
\(q(s\mid x)=\prod_y p(s_y\mid x,y)\), so \(\mathcal C(p)\)
is nonempty. The set \(\mathcal C(p)\) is closed and bounded: its constraints
are linear, and every \(q\in\mathcal C(p)\) satisfies
\(0\le q(s\mid x)\le1\). It is therefore compact. Since the objective is
continuous, the minimum is attained.
We call \(\log_2\Acup(p)\) the
\emph{union maximal leakage} and use its exponential throughout.
Allowing stochastic responses in \eqref{eq:simulation} gives the same
minimum, as Proposition~\ref{prop:lp} establishes. Unlike the communication
value of a quantum channel \cite{PR_Chitambar2023CV}, this quantity
does not optimize over physical preparations or decoding measurements.

\subsection{Compatibility under communication constraints}
Let \(B\) comprise everything the receiver holds after communication and before
the measurement setting \(y\) is chosen, including the received message, ancillary
systems, and accessible shared randomness. A measurement \(\mathsf G\)
on \(B\) produces a guess \(\widehat X\in\{1,\ldots,N\}\) of
the sender's input. For uniformly distributed \(X\), the optimal guessing
probability is
\begin{equation}
 P_g(X\mid B):=\sup_{\mathsf G}\Pr_{\mathsf G}(\widehat X=X),
 \label{eq:guessing}
\end{equation}
where the supremum runs over all allowed measurements on \(B\).
Measurements are compatible, or jointly measurable,
if they are classical postprocessings of one parent measurement. If
the receiver's measurements are compatible for each message value, the message
and the outcome of the corresponding parent measurement together define
a parent on \(B\).
Its outcome is a possible simulation output in \eqref{eq:simulation}.
Guessing \(X\) from that outcome cannot be more successful than the
optimal measurement on \(B\), which gives the information bound below.

We first state the causal assumptions for classical communication.
\begin{assumption}[One-way communication]
\label{ass:communication}
The sender's input \(x\) and the receiver's measurement setting \(y\) are independent
of each other and of the initial shared resource, including any shared
randomness. After receiving \(x\), the sender sends one classical message
\(m\in\mathcal M\), where \(\mathcal M\) is a finite nonempty message alphabet.
This is the only communication after the shared resource is distributed.
The receiver chooses \(y\) after receiving \(m\), independently of \((x,m)\)
and all systems held by the devices, and does not communicate it to the sender.
The receiver then performs the corresponding measurement, which may
depend on \(m\).
\end{assumption}
The message may be public and may be included in the reported outcome.
The alphabet bound includes all distinguishable features of the
communication. For \(k\) fixed-length bits,
\(\abs{\mathcal M}\le2^k\), with equality if all bit strings are
allowed. The message alphabet need not have a power-of-two size: three
possible messages give the threshold three, even though their binary
encoding requires two bits. Nonsignalling means that, before communication, each
party's local outcome distribution is independent of the other party's
measurement choice.

We also consider quantum communication, in which the sender sends the receiver a
quantum system instead of a classical message. Its Hilbert space has
dimension \(d\). This bound applies only to the transmitted system,
not to the initially shared systems. The input independence, timing
and exclusion of other communication in
Assumption~\ref{ass:communication} are unchanged. All states and
operations in this case are quantum.

\begin{theorem}[Compatibility bound and simulation cost]
\label{thm:communication}
\begin{enumerate}
\item[(i)] In any nonsignalling operational theory,
compatible receiver measurements satisfy
\begin{equation}
 \Acup(p)\le N P_g(X\mid B).
 \label{eq:informationbound}
\end{equation}
For the classical and quantum communication scenarios specified above,
\begin{equation}
 N P_g(X\mid B)\le
 \begin{cases}
 \abs{\mathcal M},&\text{classical message},\\
 d,&\text{quantum system, separable shared state},\\
 d^2,&\text{quantum system, arbitrary shared state}.
 \end{cases}
 \label{eq:budgets}
\end{equation}
The classical case allows arbitrary nonsignalling assistance.
\item[(ii)] For every conditional distribution
\(p\), the least message-alphabet size of a realization satisfying
Assumption~\ref{ass:communication}, with some nonsignalling shared
resource and compatible receiver measurements, is \(\lceil\Acup(p)\rceil\).
\end{enumerate}
\end{theorem}
\noindent The complete proof of both parts is given in
Appendix~\ref{app:budgets}.
A realization attaining the minimum in part~(ii) reproduces \(p(b\mid x,y)\), but may use a
different shared resource and sender encoding. If \(b\) contains a
message label, the simulator need not transmit that very label.

The freedom to use any nonsignalling resource matters. There are probabilities attainable with shared
entanglement and one classical bit that satisfy \(\Acup(p)=2\), but
require three message values for a compatible quantum realization, even
when the shared state may be chosen freely (Appendix~\ref{app:quantum-converse-gap}).

With classical communication, a violation also implies that the initial
shared quantum state is entangled. The following lemma explains why,
without requiring compatibility of the receiver's measurements.

\begin{lemma}[Classical and separable resources]
\label{lem:classicalcommunication}
Under Assumption~\ref{ass:communication}, a protocol using shared
randomness or a separable quantum state as its initial resource satisfies
\(\Acup(p)\le\abs{\mathcal M}\).
\end{lemma}

\begin{proof}
A shared-randomness protocol has the classical communication model
\[
 p(b\mid x,y)=\sum_\lambda p(\lambda)
 \sum_{m\in\mathcal M}p_A(m\mid x,\lambda)
 p_B(b\mid m,y,\lambda),
\]
where \(p(\lambda)\) is independent of \(x,y\). The hidden variable
\(\lambda\) can be taken finite, since for fixed finite alphabets these
correlations form the convex hull of finitely many deterministic protocols
\cite[Sec.~II]{PR_BraskChaves2017}.
Conditional on \((m,\lambda)\), sample an outcome independently for each
receiver setting. This gives the joint distribution
\[
 q(s\mid x)=\sum_{m,\lambda}p(\lambda)p_A(m\mid x,\lambda)
                  \prod_y p_B(s_y\mid m,y,\lambda),
\]
whose marginals are \(p(b\mid x,y)\). The maximum of a sum is at most
the sum of its maxima. Using this inequality and normalization gives
\[
 \Acup(p)\le\sum_s\max_x q(s\mid x)
 \le\sum_{m,\lambda}p(\lambda)\max_x p_A(m\mid x,\lambda)
 \le\abs{\mathcal M},
\]
since \(p_A(m\mid x,\lambda)\le1\) and \(\sum_\lambda p(\lambda)=1\).
Finite mixtures of product states have the same classical form.
Every separable state can be approximated in trace norm by such mixtures.
With the protocol fixed, each outcome probability then converges to that
of the original state. The convex hull of finitely many deterministic
protocols is closed, so the limiting distribution also has a classical
model and obeys the bound.
\end{proof}

Consequently, any exact one-way classical simulation with
shared randomness requires a message with at least
\(\lceil\Acup(p)\rceil\) possible values, or
\(\lceil\log_2\Acup(p)\rceil\) fixed-length bits.
Without preshared entanglement, transmitting a \(d\)-dimensional
quantum system instead gives the incompatibility threshold \(d\).
In each case, the communication restriction supplies an upper bound on
\(P_g(X\mid B)\), so applying \eqref{eq:informationbound} requires
neither an estimate of this guessing probability nor calibrated states.
We focus below on classical messages, whose bound follows from
nonsignalling and the number of possible messages, without a dimension
assumption on the shared quantum systems.

If the classical message is sent through a specified noisy channel
\(\mathcal N(m\mid u)\), where \(u\) is sent and \(m\) is received,
the threshold becomes the channel's communication value
\(\Lambda(\mathcal N)=\sum_m\max_u\mathcal N(m\mid u)\).
All communication must pass through this channel. Its noise is independent
of the shared resource and devices, which receive no additional information
about that noise beyond the channel output.
Proposition~\ref{prop:noisychannel} in Appendix~\ref{app:noisychannel}
states and proves this extension.

\subsection{Linear program and witnesses}
\label{sec:lp}
The minimization in \eqref{eq:operationalcosts} can be written as a linear
program (LP). Its variables are the joint outcome probabilities and upper
bounds on them, while the observed probabilities fix the marginals.
\begin{proposition}[Linear-program formulation]\label{prop:lp}
The quantity \(\Acup(p)\) is the least \(\Lambda(q)\) among all finite
classical simulations \eqref{eq:simulation}. With
\(\Omega=\prod_y\mathcal B_y\) the joint outcome space, it is the
optimum of the following LP over nonnegative variables \(q(s\mid x)\)
and \(t_s\):
\begin{equation}
\begin{aligned}
 \Acup(p)&=\min_{q,t}\sum_{s\in\Omega}t_s\\
 \textnormal{subject to}\quad
 q(s\mid x)&\le t_s &&(\forall x,s),\\
 \sum_{s\in\Omega:\,s_y=b}q(s\mid x)&=p(b\mid x,y)
                   &&(\forall x,y,b).
\end{aligned}
\label{eq:lp}
\end{equation}
\end{proposition}
\noindent The proof is given in Appendix~\ref{app:lp}.

Each \(t_s\) bounds the probability of the outcome string \(s\) for
every input. Stochastic responses can be replaced by distributions
over such strings without increasing \(\Lambda\), which gives the
reduction to this LP.
The LP has \((N+1)\abs{\Omega}\) variables. Its size grows linearly
with \(N\) for fixed measurement and outcome alphabets. For rational
data, it can be solved in time polynomial in its explicit description
length, including the bits used to represent the coefficients.

When the bound is violated, the LP also supplies a weighted sum of the
observed probabilities that certifies the violation. This linear witness is obtained from the dual LP
in Appendix~\ref{app:basic} and can be evaluated without solving the LP
again.
Each choice of real coefficients \(w_{xyb}\) satisfying the constraint
in \eqref{eq:dualcompact} defines a linear witness
\[
 W(p)=\sum_{x,y,b}w_{xyb}\,p(b\mid x,y)\le\Acup(p)
 \qquad\text{for every }p.
\]
Thus \(W(p)>\abs{\mathcal M}\) certifies incompatibility under the
classical communication bound. By strong duality, the optimal witness value
equals \(\Acup(p)\) for the given observations, so every violation
of the test has such a certificate.

The elementary properties of \(\Acup\) explain how to interpret the
test. It equals one precisely when the observations contain no
information about \(x\), and it is at least as large as the communication
value \(\sum_b\max_x p(b\mid x,y)\) for any individual measurement. The upper bound
\(\min\{N,\abs{\Omega}\}\) limits the violation possible with given
input and outcome alphabets. Operationally, a convex mixture means choosing between protocols by a
coin toss independent of \(x,y\), and pooling the data without reporting
the coin outcome. Convexity of \(\Acup\) means that this random choice
cannot create a violation when each protocol obeys the same bound. Discarding measurements or postprocessing
their outcomes independently of \(x\) cannot increase \(\Acup\) either.
These properties and the dual formula are stated and proved in
Proposition~\ref{prop:basic}. Each dual witness also defines a
communication task: with independent uniform inputs, its coefficients
can be shifted and rescaled into payoffs in \([0,1]\), preserving the
comparison between quantum and classical protocols.

With finite data, the probabilities are known only within statistical
uncertainty. If a nonempty simultaneous confidence region is given by linear
inequalities, minimizing \(\Acup(p)\) over that region remains a linear
program. A value above the communication threshold means that every
distribution in the confidence region violates the bound. A false
certificate can then occur only if the true probabilities lie outside the confidence region
(Corollary~\ref{cor:confidence}, Appendix~\ref{app:confidence}). The region
must have the claimed coverage for the sampling model used.

\section{Which fixed quantum resources can violate the bound?}
\label{sec:frontier}
Theorem~\ref{thm:communication} characterizes observations that admit
a compatible realization with unrestricted nonsignalling assistance.
We now fix the shared quantum state and the receiver's measurements.
These do not yet determine the observed probabilities, which also depend
on how the sender measures its system and encodes the message. We ask
whether some choice of these sender operations makes our test detect
incompatibility of the fixed receiver measurements. We show that a
violation is possible exactly
when two-outcome sender measurements on the same state produce
Bell-nonlocal correlations with the receiver's measurements for at least
one message value.
Here Bell-nonlocal means that the correlations
cannot be reproduced by a local hidden-variable model
\[
 P(a,b\mid x,y)=\sum_\lambda p(\lambda)
 p_A(a\mid x,\lambda)p_B(b\mid y,\lambda),
\]
with \(p(\lambda)\) independent of \(x,y\).

Let \(\rho_{AB}\) be the initial shared state and let
\(M^{(m)}=(M^{(m)}_{b|y})_{b,y}\) denote the receiver's effective POVMs conditional on
message \(m\), for \(\abs{\mathcal M}\ge2\). The index \(b\) denotes
the measurement outcome. These POVMs act on the receiver's initial system \(B\)
and include all local processing and fixed ancillary systems. Finite
outcome alphabets \(\mathcal B_y\) may be chosen to agree for all
messages by adding zero effects. The receiver may classically process
\((m,b,y)\), independently of \(x\). In particular, the reported outcome may be the
pair \((m,b)\), in which case the sender's POVMs \((A_{m|x})_m\) give
\[
 p(m,b\mid x,y)=\tr[(A_{m|x}\otimes M^{(m)}_{b|y})\rho_{AB}].
\]
The reported alphabet is then \(\mathcal M\times\mathcal B_y\).
Reporting a message already received does not increase communication.
The state \(\rho_{AB}\) includes any accessible classical registers
correlated with the quantum systems.

\Needspace{5\baselineskip}
\begin{theorem}[Characterization for fixed quantum resources]\label{thm:operationalfrontier}
Under Assumption~\ref{ass:communication}, fix a shared state
\(\rho_{AB}\) and the receiver's measurement families
\(M^{(m)}=(M^{(m)}_{b|y})_{b,y}\), one for each message value
\(m\in\mathcal M\), with \(\abs{\mathcal M}\ge2\).
The finite sender-input and reported-outcome alphabets, the sender's
measurements and the receiver's classical output processing may be chosen.
The following are equivalent:
\begin{enumerate}
 \item A communication experiment using the fixed state and receiver
       measurements can produce probabilities \(p\) with
       \(\Acup(p)>\abs{\mathcal M}\), certifying incompatibility of
       at least one measurement family \(M^{(m)}\).
 \item For some message value \(m\), the receiver's measurements
       \(M^{(m)}\) and a finite family of two-outcome sender POVMs
       \(F_{a|x}\) produce Bell-nonlocal correlations
       \[
       P(a,b\mid x,y)=\tr[(F_{a|x}\otimes M^{(m)}_{b|y})\rho_{AB}],
       \qquad a\in\{0,1\}.
       \]
\end{enumerate}
Whenever these conditions hold, a one-bit protocol achieves
\(\Acup(p)>2\) with the same state, using the receiver's measurements
for the message value in condition~2 and for one other message value.
\end{theorem}
\noindent The proof is given in Appendix~\ref{app:frontier-proof}.

In condition~2, the label \(m\) selects the receiver's measurements in advance,
and no message is transmitted. It is not chosen from the sender's
outcome in this Bell experiment.
The two-outcome restriction concerns the sender's tests, not the number
of settings or outcomes available to the receiver.

The proof relates common classical simulations of the communication
experiment to local models of the Bell correlations. If a local model
exists for each fixed message value, these models combine into a common
classical simulation \(q\) with \(\Lambda(q)\le\abs{\mathcal M}\).
The observed probabilities therefore satisfy the compatibility bound.
Conversely, we add inputs on which the sender always sends a specified
message and let the receiver include that message in the report.
A joint explanation satisfying the compatibility bound would then give a local
model for the original Bell correlation. Additional sender inputs may be
necessary for this construction: \(\Acup\le N\) rules out a violation
of our one-bit bound with only two inputs, even when the associated Bell
correlation violates the Clauser--Horne--Shimony--Holt (CHSH) inequality.

The use of deterministic additional inputs has precedents in
Bell-to-instrumental constructions \cite[Secs.~5--6]{PR_VanHimbeeck2019}
and in the nonsignalling reconstruction of
Vieira \emph{et al.}~\cite[Theorem~1]{PR_Vieira2025Outcome}.
Here it gives an exact detection criterion for a fixed state and
receiver measurements. A local model
for all sender and receiver measurements already
implies a classical simulation \cite{PR_Pauwels2022}; the theorem requires
locality only separately for the measurements associated with each message.

If the state may also be chosen, the receiver's measurement families are undetectable
exactly when each \(M^{(m)}\) is Bell-local for every shared state and
all two-outcome sender POVMs. For noisy projective qubit measurements,
a related result equates classical one-bit simulability for all
preparations, with shared randomness, to Bell locality against two-outcome
sender measurements on a maximally entangled state
\cite[Theorem~3]{PR_Egelhaaf2025}.

Every incompatible binary pair can
violate CHSH on a suitable shared state
\cite[Proposition~2]{PR_Wolf2009}. Using such a pair for one message
value therefore permits a one-bit violation when that state is available.
Detection is not restricted to pairs: for every \(n\ge3\), Bene and
V\'ertesi construct \(n\) binary measurements whose proper subfamilies
are all compatible, but which are Bell-nonlocal against two-outcome
measurements \cite[Sec.~IV]{PR_Bene2018}. Using such a family for one
message value permits detection of global incompatibility. Nevertheless,
some incompatible families remain undetectable: a known family of three
incompatible qubit measurements is Bell-local for every shared state
\cite[Secs.~II--III]{PR_Bene2018}. When used for every message value,
this family cannot violate the bound for any sender measurements and
classical output processing allowed by the theorem.

\Needspace{10\baselineskip}
\section{Private randomness with classical communication}
\label{sec:privacy}

We apply the incompatibility test to private randomness generation at
the receiver, allowing the communicating sender to be adversarial.
The sender may retain quantum information correlated with the receiver's
outputs. We first bound their conditional entropy in one round, then
use a fixed linear test to certify extraction of private bits from
repeated rounds with device memory. Neither step requires reconstructing
the state or measurements.

\subsection{A lower bound on conditional output entropy}

In this section, \(X,Y,M,B\) denote the classical random inputs, message
and outcome, with values \(x,y,m,b\), respectively. In particular,
\(B\) now denotes the classical output rather than the receiver's quantum
system.
The sender and receiver may initially share an arbitrary finite-dimensional
quantum state, with no bound on its dimension. A trusted classical controller
supplies independent uniform inputs
$X\in\{1,\ldots,N\}$ and $Y\in\{1,\ldots,n\}$, with $n\ge2$.
The inputs are independent of all prior device and adversary systems.
The sender transmits $M\in\mathcal M$ before $Y$ is chosen, and the
receiver returns a binary outcome $B\in\{0,1\}$ privately to the controller
(Fig.~\ref{fig:privacy}). The message is the only communication between
the devices, and no other signal carrying information about $X$ reaches
the receiver. The controller records the message together with the outcome.

The adversary's system $E$ includes the sender's remaining quantum
system and the quantum systems held by any external adversary. The adversary may postpone measuring
$E$ until after the inputs and message are revealed.
No receiver output or output-dependent feedback is released during data
generation. All systems retained by the receiver, including any stored
outcomes, remain isolated afterwards or are securely erased. Otherwise,
a device with memory could store its outputs and reveal them later.
Both devices may retain internal memory during the experiment.

\begin{figure}[ht]
\centering
\begin{tikzpicture}[x=1cm,y=1cm]
 \node[paper box,minimum width=23mm] (sender) at (2.3,0) {sender};
 \node[paper box,minimum width=23mm] (receiver) at (8.8,0) {receiver};
 \node[paper box,minimum width=28mm] (state) at (5.55,1.25) {shared state};
 \draw[paper arrow] (state.west) -| (sender.north);
 \draw[paper arrow] (state.east) -| ([xshift=-6mm]receiver.north);
 \node[paper label] (x) at (.15,0) {$x$};
 \draw[paper classical,paper arrow] (x.east) -- (sender.west);
 \draw[paper classical,paper arrow] (sender.east) -- (receiver.west)
       node[midway,above=5pt,paper label] {message $m\in\mathcal M$};
 \node[paper label] (y) at (9.4,1.25) {$y$};
 \draw[paper classical,paper arrow] (y.south) -- ([xshift=6mm]receiver.north);
 \node[paper box,minimum width=37mm] (process) at (8.8,-1.45) {trusted processing};
 \draw[paper classical,paper arrow] (5.55,0) -- (5.55,-1.45) --
       node[midway,above=3pt,paper label] {$m$} (process.west);
 \draw[paper classical,paper arrow] (receiver.south) --
       node[right=3pt,paper label] {$b$} (process.north);
 \draw[paper classical,paper arrow] (process.east) -- (12.1,-1.45)
       node[right,paper label] {$K$};
 \node[paper label] (eve) at (2.3,-1.45) {$E$};
 \draw[paper arrow] (sender.south) -- (eve.north);
\end{tikzpicture}
\caption{Private randomness with a public incoming message. The input $y$
is chosen after $m$ is fixed. The controller records a copy of $m$, collects
$b$ privately, and extracts $K$ after testing the statistics. The system $E$ contains the
sender's residual and external quantum information. During data generation,
there is no return communication to the sender. Double lines carry classical information.}
\label{fig:privacy}
\end{figure}
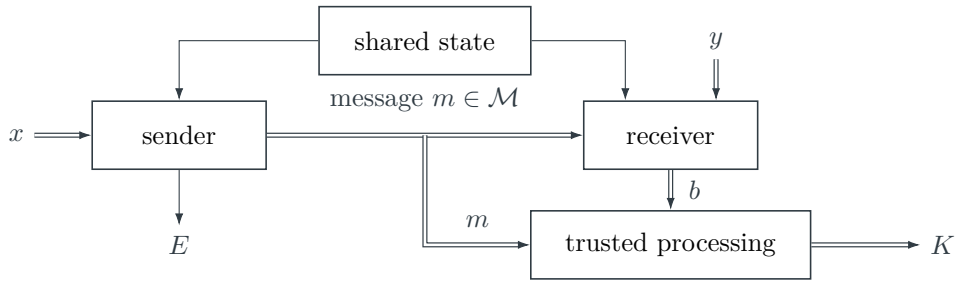

\Needspace{5\baselineskip}
The conditional entropy $H(B\mid EXYM)$ measures the uncertainty about
$B$ that remains when the adversary holds $E$ and learns the inputs and
message. It is an average over the uniformly sampled inputs and the
resulting messages, so a positive bound need not hold for every fixed
triple $(x,y,m)$. The following theorem bounds this uncertainty from the violation of
the union-leakage bound.

\begin{theorem}[Private output entropy]
\label{thm:privacy}
In the quantum scenario of Section~\ref{sec:privacy}, let
$p=p(m,b\mid x,y)$ be the joint probabilities of the actual message and
receiver outcome. Put
$\delta=[\Acup(p)-\abs{\mathcal M}]_+$.
Then
\begin{equation}
 H(B\mid EXYM)\ge
 1-h_2\!\left(\frac12+\frac{2\delta}{Nn(n-1)}\right).
 \label{eq:privacy-entropy}
\end{equation}
Here $H$ is conditional von Neumann entropy and $h_2$ is binary entropy.
In particular, $\Acup(p)>\abs{\mathcal M}$ certifies positive conditional
output entropy.
\end{theorem}
\noindent The proof is given in Appendix~\ref{app:privacy-entropy}.

Forgetting the message cannot increase $\Acup$, so one may also use
$[\Acup(p_B)-\abs{\mathcal M}]_+$, where
$p_B(b\mid x,y)=\sum_m p(m,b\mid x,y)$.
The entropy is still conditioned on the public message. The raw output
need not be uniform; extraction from many outputs is the next step.

The proof compares the experiment with a hypothetical measurement
that performs all receiver measurements sequentially and records their
outcomes. This defines a parent measurement for each message and therefore
obeys the compatibility bound. A violation forces disturbance in this
sequential construction, which the proof relates to the remaining output
uncertainty when the adversary holds a purification.

\subsection{Extracting private bits over repeated rounds}
\label{sec:privacy-sequential}

The quantity that controls extraction from a finite output string is
its smooth conditional min-entropy. Without smoothing, conditional
min-entropy is the negative logarithm of the adversary's optimal
probability of guessing the entire string given its side information.
Smoothing allows a small controlled approximation of the state.
Generalised entropy accumulation \cite{PR_Metger2024} relates our
single-round von Neumann entropy bound to this quantity for the complete
sequence, even when device memory correlates different rounds.
Privacy amplification against quantum side information
\cite{PR_Tomamichel2011} then turns the outputs into nearly uniform private
bits: a hash function chosen with independent public randomness maps the
output sequence to a shorter bit string.

For the privacy protocol, we subtract the compatibility bound from a
linear witness, so that positive values certify a violation. The
constant subtraction can be included in its coefficients because the
conditional probabilities sum to one for each input pair. In this
section, $W$ denotes this shifted witness. Fix real coefficients
$w_{xymb}$ satisfying
\begin{equation}
 W(p)=\sum_{x,y,m,b}w_{xymb}\,p(m,b\mid x,y)
 \le\Acup(p)-\abs{\mathcal M}
 \qquad\text{for every }p.
 \label{eq:privacy-witness}
\end{equation}
The dual of the union-leakage linear program supplies these coefficients
(Appendix~\ref{app:basic}). Every violation of the compatibility bound has such a witness
with $W(p)>0$.

Theorem~\ref{thm:privacy} applies with $\delta=[W(p)]_+$.
For uniform inputs, the round score $S=Nn\,w_{XYMB}$ has expectation
$W(p)$, so its observed average is the test statistic. The factor $Nn$
accounts for the uniform sampling of the input pair. With device memory,
the round distributions may differ and the scores may be correlated.
Entropy accumulation bounds the entropy of the whole output string
from the accepted average score.

\paragraph{Protocol.}
Fix the coefficients before collecting data. The controller then
proceeds as follows.
\begin{enumerate}
\item Choose the number of rounds $L$, a security parameter
$0<\varepsilon<1$, and an acceptance threshold
\[
 0<w_0<\frac{Nn(n-1)}4.
\]
\item Run $L$ rounds of the scenario above, with independent uniform
inputs in every round and the same communication and isolation
conditions throughout. Record $S_i=Nn\,w_{X_iY_iM_iB_i}$ privately.
Accept if $L^{-1}\sum_{i=1}^L S_i\ge w_0$, and otherwise abort.
\item On acceptance, choose a hash function uniformly from a two-universal family
using an independent public random seed sampled after data generation.
Apply it to $B_1\cdots B_L$ to produce a string $K$. The output length
$\ell$ is fixed in advance. The corollary below gives a guaranteed
asymptotic rate for this extraction.
\end{enumerate}
Two-universality means that any two distinct input strings have
probability at most $2^{-\ell}$ of producing the same hash output under
the random choice of function.
The acceptance rule uses the mean of the actual round scores. Evaluating
$W$ on empirical conditional probabilities normalized separately for
each input pair generally gives a different statistic. The guarantee
below also requires the coefficients to be fixed before these rounds,
not chosen retrospectively from the scored data.

The public transcript $\mathcal T$ contains all inputs and messages,
the acceptance flag and the hash seed. It may also contain the final
counts of each tuple $(x,y,m,b)$, or a function of those counts.
Individual outcomes and scores are not disclosed. The aggregate counts
still reveal information about the outputs, which is accounted for in
the finite-size correction below.

\Needspace{10\baselineskip}
\begin{corollary}[Composable private randomness]
\label{cor:privacy-extraction}
For the protocol in Section~\ref{sec:privacy-sequential}, including its communication, isolation and
disclosure conditions, and for all sufficiently large $L$, there is a
fixed output length $\ell$ satisfying
\[
 \begin{aligned}
  \ell&\ge Lr(w_0)-O(\sqrt L),\\
  r(w_0)&=1-h_2\!\left(\frac12+
       \frac{2w_0}{Nn(n-1)}\right).
 \end{aligned}
\]
The constant in $O(\sqrt L)$ depends only on the fixed alphabets,
coefficients, acceptance threshold and $\varepsilon$, not on
the device dimensions.

Let $\sigma^{\mathrm{acc}}_{KE\mathcal T}$ be the subnormalized final
state on acceptance, whose trace is the acceptance probability, and
let $\sigma^{\mathrm{acc}}_{E\mathcal T}
=\tr_K\sigma^{\mathrm{acc}}_{KE\mathcal T}$.
The extracted string satisfies
\begin{equation}
 \frac12
 \left\|\sigma^{\mathrm{acc}}_{KE\mathcal T}
 -\frac{\id_K}{2^\ell}\otimes\sigma^{\mathrm{acc}}_{E\mathcal T}\right\|_1
 \le\varepsilon.
 \label{eq:privacy-security}
\end{equation}
\end{corollary}
\noindent The proof is given in Appendix~\ref{app:privacy-composability}.

The rate $r(w_0)$ is a lower bound on the asymptotic output length
per round in accepted runs. Choosing a numerical output length
for a given $L$ also requires evaluating the finite-size correction.

The ideal state in \eqref{eq:privacy-security} contains a uniform
string independent of the adversary and the public transcript.
Subnormalization weights the distance by the probability of acceptance.
For the complete protocol, including abort, replacing the extracted
string by independent uniform bits changes the probability of any event
obtained by subsequent processing of $K,E,\mathcal T$ by at most
$\varepsilon$. The receiver's memory and unreleased raw data must remain
inaccessible to the adversary. The inputs and hash seed require independent
randomness, so this is not a claim of net randomness expansion.

In quantum key distribution with an untrusted receiver, compatibility
of the effective receiver measurements, including the preceding channel,
prevents the observed statistics alone from certifying a secret key
\cite[Sec.~5]{PR_Masini2024}. The private randomness guarantee proved here
does not establish a shared key, which would require a separate protocol
and security analysis.

\section{Conclusion and outlook}
\label{sec:outlook}
Union maximal leakage tests measurement incompatibility in experiments
with bounded classical communication. A linear program evaluates its
exponential from the observed probabilities, and this value determines the minimum
number of possible message values needed for a compatible explanation
with unrestricted nonsignalling assistance.

For a fixed quantum state and receiver measurements, a violation is
possible exactly when two-outcome sender measurements can produce
Bell nonlocality with one of the receiver's measurement families. This characterizes
both the incompatibility detected by the test and its limitations, under
the operations specified in Theorem~\ref{thm:operationalfrontier}.

In quantum theory, for binary receiver outcomes, the violation also
bounds the conditional output entropy against the communicating sender and any external
adversary. Together with entropy accumulation and privacy amplification,
this gives composable randomness extraction for devices with memory.

The exact simulation result permits arbitrary nonsignalling resources.
Restricting the compatible explanation to quantum resources can require
more communication, even for probabilities that admit a quantum realization
(Appendix~\ref{app:quantum-converse-gap}). Incorporating this restriction
into a tractable stronger test is an open question.

Semidefinite programming methods for bounding conditional entropy
\cite{PR_BrownFawzi2024}, combined with relaxations for sequential quantum
correlations \cite{PR_Bowles2020}, offer a possible way to use the full
observed distribution for privacy certification. A related question is
which features of these data, beyond \(\Acup(p)\), can strengthen the
entropy guarantee.

\section*{Acknowledgements}
This work grew out of the CQT Hackamonth 2026. The author thanks
Valerio Scarani for suggesting a different problem that led to this work.
She is grateful to the Centre
for Quantum Technologies, Singapore, for its hospitality and to the
event organizers.

She also thanks Ritu Dhaulakhandi for helpful discussions on the test's
cryptographic content, Otfried G{\"u}hne for suggesting the connection
with violations of Bell inequalities, Andreas Winter for discussions on
possible applications, Christoph Hirche for comments on the final
manuscript, and Rutvij Bhavsar, Tim Achenbach, and Mu-En Liu for interesting discussions.

During this work, the author was supported by the Federal Ministry of
Research, Technology and Space (BMFTR), Germany, through the QC service
center QUICS (grant no.~13N17418).

\section*{Disclosure of AI usage}
OpenAI's ChatGPT-6, accessed through ChatGPT in September 2026, was used
to substantially improve the author's initial work. In particular, it was
used to search for prior work, generalize statements interactively,
identify applications of the test, and assist in preparing the manuscript.
The research idea originated with the author, who takes responsibility
for the paper's content.
\printbibliography

\clearpage
\appendix
\section*{Appendix}
Statements from the main text are recalled below with their original
numbers and titles.

\section{Guessing bounds and linear programming}
\label{app:test}
\subsection{Compatibility bound and simulation cost}
\label{app:budgets}

\begin{recalled}{Theorem}{thm:communication}{Compatibility bound and simulation cost}
\begin{enumerate}
\item[(i)] In any nonsignalling operational theory,
compatible receiver measurements satisfy
\[
 \Acup(p)\le N P_g(X\mid B).
 \tag{\ref*{eq:informationbound}}
 \]
For the classical and quantum communication scenarios specified above,
\[
 N P_g(X\mid B)\le
 \begin{cases}
 \abs{\mathcal M},&\text{classical message},\\
 d,&\text{quantum system, separable shared state},\\
 d^2,&\text{quantum system, arbitrary shared state}.
 \end{cases}
 \tag{\ref*{eq:budgets}}
 \]
The classical case allows arbitrary nonsignalling assistance.
\item[(ii)] For every conditional distribution
\(p\), the least message-alphabet size of a realization satisfying
Assumption~\ref{ass:communication}, with some nonsignalling shared
resource and compatible receiver measurements, is \(\lceil\Acup(p)\rceil\).
\end{enumerate}
\end{recalled}

\begin{proof}[Proof of Theorem~\ref{thm:communication}]
\emph{Part (i): from a parent measurement to a guessing strategy.}
Conditional on each outcome of the parent measurement, independently
sample the classical response for every setting and collect the results in \(s\in\Omega\).
This defines an allowed measurement on \(B\), whose distribution
\(q(s\mid x)\) has the observed marginals and hence belongs to
\(\mathcal C(p)\). Assigning each \(s\) to an input maximizing
\(q(s\mid x)\) gives a valid guessing strategy, so
\[
 \Acup(p)\le\sum_s\max_x q(s\mid x)\le N P_g(X\mid B).
\]
It remains to bound the receiver's probability of guessing the whole input
\(X\). The following bounds hold independently of whether these
measurements are compatible.

\emph{Classical messages.} Fix a guessing measurement for each
candidate message \(m\). Before communication, these choices define a
nonsignalling correlation \(P(a,g\mid x,m)\), where \(a\) is the sender's
message outcome and \(g\in\{1,\ldots,N\}\) is the receiver's guess. Here \(m\)
is an independently chosen measurement label. Nonsignalling gives
\(\sum_aP(a,g\mid x,m)=P_B(g\mid m)\), independently of \(x\).
In the actual protocol, the receiver uses the received value \(a\) to select
the measurement \(m=a\). This guessing strategy \(\mathsf G\) satisfies
\[
 N\Pr_{\mathsf G}(\widehat X=X)
 =\sum_{x,m}P(a=m,g=x\mid x,m)
 \le\sum_{x,m}P_B(g=x\mid m)=\abs{\mathcal M}.
\]
The inequality bounds each joint event by the receiver's marginal, obtained by
summing over all of the sender's outcomes.
Taking the supremum over guessing measurements proves the first bound.
This includes any shared randomness in the nonsignalling resource.

\emph{Quantum messages.} For a transmitted system \(C\) with
\(\dim\mathcal H_C=d\), Proposition~1 of
Ref.~\cite{PR_Tavakoli2021EA} bounds the optimal probability of guessing
the sender's input after transmission. For a uniform input with \(N\)
possible values, it gives \(P_g\le d/N\) without initial shared
entanglement and \(P_g\le d^2/N\) with arbitrary shared entanglement.
These are the two quantum cases of \eqref{eq:budgets}. We give the
operator argument to make their role explicit.
Write \(B_0\) for the receiver's initial quantum system and
\(\tau^x_{CB_0}\) for its joint state with the transmitted system \(C\).
Its marginal
\(\tr_C\tau^x_{CB_0}=\rho_{B_0}\) is independent of \(x\).
Every pure state on \(CB_0\) has a Schmidt decomposition
\[
 \ket{\psi}=\sum_{j=1}^r\sqrt{\lambda_j}\,
 \ket{u_jv_j},\qquad r\le d.
\]
Cauchy--Schwarz gives
\[
 \ketbra{\psi}
 \preceq r\sum_{j=1}^r\lambda_j\ketbra{u_jv_j}
 \preceq d\,\id_C\otimes\tr_C\ketbra{\psi}.
\]
Applying this inequality to pure-state decompositions and taking
trace-norm limits gives
\(\tau^x_{CB_0}\preceq d\,\id_C\otimes\rho_{B_0}\),
also when \(B_0\) is infinite-dimensional.
Without initial shared entanglement, \(\tau^x_{CB_0}\) is separable.
Bounding each normalized \(C\)-factor in a product-state mixture by
\(\id_C\) gives the stronger inequality
\(\tau^x_{CB_0}\preceq\id_C\otimes\rho_{B_0}\);
trace-norm continuity extends it to arbitrary separable states.

For any guessing POVM \((G_x)_x\), summing and using
\(\sum_xG_x=\id_C\otimes\id_{B_0}\) gives
\[
 \sum_x\tr(G_x\tau^x_{CB_0})
 \le d\,\tr\!\left[
       \Bigl(\sum_xG_x\Bigr)(\id_C\otimes\rho_{B_0})
       \right]=d^2.
\]
The separable bound gives \(d\) by the same calculation.
Subsequent local processing cannot improve guessing, so taking the
supremum proves \eqref{eq:budgets}.

\emph{Part (ii): exact simulation.} Choose an optimal joint distribution \(q\in\mathcal C(p)\). Cubitt
\emph{et al.}'s channel-simulation theorem realizes \(q(s\mid x)\)
using a classical message with
\(\lceil\Lambda(q)\rceil=\lceil\Acup(p)\rceil\) possible values and
nonsignalling assistance \cite[Theorem~16]{PR_Cubitt2011}.
The receiver obtains \(s\) before choosing \(y\) and reports \(s_y\), so the
resulting measurements are compatible. The lower bound follows from
part~(i).
\end{proof}

\Needspace{14\baselineskip}
\subsection{Derivation of the linear program}
\label{app:lp}
\begin{recalled}{Proposition}{prop:lp}{Linear-program formulation}
The quantity \(\Acup(p)\) is the least \(\Lambda(q)\) among all finite
classical simulations \eqref{eq:simulation}. With
\(\Omega=\prod_y\mathcal B_y\) the joint outcome space, it is the
optimum of the following LP over nonnegative variables \(q(s\mid x)\)
and \(t_s\):
\[
\begin{aligned}
 \Acup(p)&=\min_{q,t}\sum_{s\in\Omega}t_s\\
 \textnormal{subject to}\quad
 q(s\mid x)&\le t_s &&(\forall x,s),\\
 \sum_{s\in\Omega:\,s_y=b}q(s\mid x)&=p(b\mid x,y)
                   &&(\forall x,y,b).
\end{aligned}
\]
\end{recalled}

\begin{proof}[Proof of Proposition~\ref{prop:lp}]
For a classical channel \(T(s\mid z)\), nonnegativity and normalization give
\begin{equation}
 \sum_s\max_x\sum_z T(s\mid z)q(z\mid x)
 \le\sum_{s,z}T(s\mid z)\max_xq(z\mid x)
 =\Lambda(q).
 \label{eq:dataprocessing}
\end{equation}
Take \(T(s\mid z)=\prod_y p_B(s_y\mid y,z)\).
Then \(\widetilde q(s\mid x)=\sum_zT(s\mid z)q(z\mid x)\) is normalized,
and summing over all coordinates except \(s_y\) recovers
\eqref{eq:simulation}, so \(\widetilde q\in\mathcal C(p)\) and
\(\Lambda(\widetilde q)\le\Lambda(q)\) by \eqref{eq:dataprocessing}.
Conversely, a joint distribution in \(\mathcal C(p)\) defines a simulation with deterministic responses
\(p_B(b\mid y,s)=\mathbf 1\{b=s_y\}\).
Finally, for fixed \(q\) the smallest feasible value of \(t_s\) is
\(\max_xq(s\mid x)\), giving the LP.
\end{proof}

\subsection{Properties and linear witnesses}
\label{app:basic}
\begin{proposition}[Basic properties and linear witnesses]
\label{prop:basic}
\begin{enumerate}
\item[(i)] Independently of the physical realization, every conditional
distribution \(p(b\mid x,y)\) satisfies
\begin{equation}
 1\le\max_y\sum_b\max_x p(b\mid x,y)
 \le\Acup(p)\le\min\{N,\abs{\Omega}\}.
 \label{eq:basicbounds}
\end{equation}
For \(n=1\), \(\Acup(p)=\sum_b\max_xp(b\mid x,1)\).
Moreover,
\[
 \Acup(p)=1\quad\Longleftrightarrow\quad
 p(b\mid x,y)=p(b\mid x',y)\quad\text{for all }x,x',y,b.
\]
\item[(ii)] If \(p'\) is obtained from \(p\) by discarding measurements or
postprocessing their outcomes independently of \(x\), then
\[
 \Acup(p')\le\Acup(p).
\]
\item[(iii)] The dual representation is
\begin{equation}
\begin{aligned}
 \Acup(p)=\max_w\quad
   &\sum_{x,y,b}w_{xyb}\,p(b\mid x,y)\\[0.5ex]
 \textnormal{subject to}\quad
   &\sum_{x=1}^N\left[\sum_{y=1}^n w_{xy\,s_y}\right]_+
       \le1\qquad\textnormal{for every }s\in\Omega,
\end{aligned}
\label{eq:dualcompact}
\end{equation}
where \(w_{xyb}\in\mathbb R\).
Each feasible \(w\) defines a linear witness bounded by \(\Acup\)
for every conditional distribution \(p\), and the maximum is attained.
In particular, \(\Acup\) is convex: for conditional distributions \(p,p'\) with the same alphabets and
\(0\le\theta\le1\),
\[
 \Acup\bigl(\theta p+(1-\theta)p'\bigr)
 \le\theta\Acup(p)+(1-\theta)\Acup(p').
\]
\end{enumerate}
\end{proposition}
\begin{proof}[Proof of Proposition~\ref{prop:basic}]
\emph{(i)} Taking a marginal of a joint probability distribution and applying
\eqref{eq:dataprocessing} gives the single-measurement lower bound, which is
at least one by normalization. Each maximum \(\max_x q(s\mid x)\)
is at most one and at most \(\sum_x q(s\mid x)\). Summing over \(s\)
gives the upper bounds \(\abs{\Omega}\) and \(N\), respectively.
For one measurement, the joint distribution is the observed channel itself.
If the optimum is one, take a feasible \((q,t)\) with \(\sum_s t_s=1\).
For every \(x\), the entries of \(t-q(\cdot\mid x)\) are nonnegative
and sum to zero, forcing \(q(\cdot\mid x)=t\). All observed probabilities
are independent of \(x\).
Conversely, if \(p(b\mid x,y)\) is independent of \(x\), the same joint distribution
can be used for every input, giving \(\Lambda(q)=1\).

\emph{(ii)} Deleting or postprocessing measurements changes only the responses
\(p_B(b\mid y,z)\) in \eqref{eq:simulation}, leaving the encoding \(q\) and
\(\Lambda(q)\) unchanged. Minimizing over simulations gives the inequality.

\emph{(iii)} Dualize the marginal equalities in \eqref{eq:lp} with
real multipliers \(w_{xyb}\), retaining \(0\le q(s\mid x)\le t_s\)
in the optimization domain. For fixed \(t\), minimizing the Lagrangian
over \(q\) gives
\[
 \sum_{x,y,b}w_{xyb}p(b\mid x,y)
 +\sum_s t_s\left(1-\sum_x\left[\sum_y w_{xy\,s_y}\right]_+\right).
\]
The coefficient of \(q(s\mid x)\) is \(-\sum_yw_{xy\,s_y}\).
Its minimizing value is \(t_s\) when this coefficient is negative and
zero when it is positive. When it is zero, either choice is optimal.
This gives the positive parts above. The infimum over \(t\ge0\) is
finite exactly under the constraints in \eqref{eq:dualcompact}. Weak duality gives the witness
bound. The primal LP is feasible and bounded, so strong LP duality gives
equality and attainment.
Convexity follows because \(\Acup\) is a maximum of linear functionals
over a feasible set independent of \(p\).
\end{proof}
The upper bound in part~(i) uses only the input and outcome alphabets;
it applies even to distributions that no allowed experiment can realize.
Theorem~\ref{thm:communication} instead uses the physical communication
constraint and compatibility. For example, with four sender inputs and
two binary receiver measurements, part~(i) gives \(\Acup\le4\), while
compatible measurements with one communicated bit obey \(\Acup\le2\).

For the mixture in part~(iii), an input-independent coin chooses the
protocol producing \(p\) with probability \(\theta\), and the one
producing \(p'\) otherwise. Pooling the trials gives
\(\theta p+(1-\theta)p'\). Reporting the coin outcome would instead
give a different distribution with an enlarged outcome alphabet.
The convexity statement concerns \(\Acup\), the exponential of union
maximal leakage. It does not assert convexity of \(\log_2\Acup\).

Normalization allows constants to be absorbed into the witness coefficients.
For the uniform two-bit random-access task, write $x=(x_1,x_2)\in\{0,1\}^2$
and $y\in\{1,2\}$, with binary outcome $b$. The coefficients
$w_{xyb}=\mathbf1\{b=x_y\}-1/2$ satisfy \eqref{eq:dualcompact}:
for each $s$, only $x=s$ gives a positive term, equal to one. Hence
\[
 8p_{\mathrm{succ}}-4\le\Acup(p),\qquad
 p_{\mathrm{succ}}:=\frac18\sum_{x,y}p(b=x_y\mid x,y).
\]
A success probability above $3/4$ therefore violates the one-bit bound.

\subsection{Communication through a noisy classical channel}\label{app:noisychannel}
Replace the noiseless message by a single use of a specified classical
channel with transition probabilities \(\mathcal N(m\mid u)\)
\cite[Sec.~II]{PR_Cubitt2011}. The sender supplies input \(u\).
Only the channel output \(m\) is delivered to the receiver before the
measurement setting is chosen.
The channel's noise is independent of the shared resource and device
randomness, and neither device receives additional side information about that noise.
All input-dependent communication passes through this channel.

\begin{proposition}[Noisy-channel bound]\label{prop:noisychannel}
In any nonsignalling operational theory, compatible receiver measurements
for each received message \(m\) satisfy
\begin{equation}
\Acup(p)\le \Lambda(\mathcal N)=\sum_m\max_u \mathcal N(m\mid u).
 \label{eq:noisychannel}
\end{equation}
In quantum theory, the same bound holds for separable resources and
arbitrary measurements.
For independent channels,
\(\Lambda(\mathcal N_1\otimes \mathcal N_2)=\Lambda(\mathcal N_1)\Lambda(\mathcal N_2)\).
\end{proposition}
\begin{proof}
For each possible channel output \(m\), fix a guessing measurement
on the receiver's systems. Before communication, these choices define a
nonsignalling correlation \(P(u,g\mid x,m)\) between the sender's channel
input \(u\) and the receiver's guess \(g\). Here \(m\) is an independently
chosen measurement label, not yet the channel output.
Nonsignalling gives
\(\sum_uP(u,g\mid x,m)=P_B(g\mid m)\), independently of \(x\).
In the actual protocol, the channel output selects the receiver's measurement.
For this guessing strategy \(\mathsf G\), independence of the channel
noise yields
\[
\begin{aligned}
 N\Pr_{\mathsf G}(\widehat X=X)
 &=\sum_{x,m,u}\mathcal N(m\mid u)P(u,g=x\mid x,m)\\
 &\le\sum_m\bigl(\max_u \mathcal N(m\mid u)\bigr)\sum_x P_B(g=x\mid m)
 =\Lambda(\mathcal N).
\end{aligned}
\]
The inequality replaces each channel probability by its largest value
at fixed \(m\), and each \(P_B(\cdot\mid m)\) is normalized. Taking the supremum
over guessing measurements gives \(N P_g(X\mid B)\le\Lambda(\mathcal N)\).
Theorem~\ref{thm:communication}(i) therefore proves the bound for
compatible receiver measurements.

For a finite mixture of product states, the state label \(\lambda\)
gives a shared-randomness model with local responses
\(p_A(u\mid x,\lambda)\) and \(p_B(b\mid m,y,\lambda)\).
The distribution \(p(\lambda)\) is independent of the inputs.
The simulation channel with output \(z=(\lambda,m)\) obeys
\[
 q(\lambda,m\mid x)=p(\lambda)
       \sum_u\mathcal N(m\mid u)p_A(u\mid x,\lambda)
 \le p(\lambda)\max_u\mathcal N(m\mid u).
\]
It simulates arbitrary receiver measurements and satisfies
\(\Lambda(q)\le\Lambda(\mathcal N)\), which bounds \(\Acup(p)\) by
Proposition~\ref{prop:lp}. An arbitrary separable state is a trace-norm
limit of finite mixtures of product states. The resulting probabilities
converge. Every linear witness in Proposition~\ref{prop:basic}(iii)
remains bounded by \(\Lambda(\mathcal N)\) in the limit; maximizing over
these witnesses extends the bound to the limiting probabilities.
Finally, the maxima for a product channel separate:
\[
 \sum_{m_1,m_2}\max_{u_1,u_2}
 \mathcal N_1(m_1\mid u_1)\mathcal N_2(m_2\mid u_2)
 =\left(\sum_{m_1}\max_{u_1}\mathcal N_1(m_1\mid u_1)\right)
  \left(\sum_{m_2}\max_{u_2}\mathcal N_2(m_2\mid u_2)\right).
\]
\end{proof}
The communication value \(\Lambda(\mathcal N)\) depends on the transition
probabilities of the physical channel \cite{PR_Chitambar2023CV}.
For a noiseless channel with \(2^k\) symbols, \(\Lambda(\mathcal N)=2^k\).
A binary symmetric channel with crossover probability \(0\le\varepsilon\le1\)
has \(\Lambda(\mathcal N)=1+\abs{1-2\varepsilon}\), so an ideal \(k\)-bit
message together with an independent binary symmetric side channel gives
the threshold \(2^k(1+\abs{1-2\varepsilon})\).
The bound requires a model of the physical channel, including its noise.
The observed conditional distribution of the received message given
\(x\) does not by itself specify this model. For example, if
\(m=x\oplus\lambda\) for a binary input \(x\) and a uniformly random
shared bit \(\lambda\), the message is independent of \(x\), but the
receiver can recover \(x\) using \(\lambda\).

The role of contextuality in entanglement-assisted communication
through noisy classical channels is studied in
Ref.~\cite{PR_YadavalliKunjwal2022}.

\Needspace{8\baselineskip}
\subsection{Confidence regions}
\label{app:confidence}
\begin{corollary}[Confidence-region certificate]
\label{cor:confidence}
Let \(p_{\mathrm{true}}(b\mid x,y)\) denote the experiment's unknown
conditional outcome probabilities, and let \(0<\alpha<1\).
Suppose the data determine a polytope \(\mathcal P\) of conditional
distributions with simultaneous coverage
\(\Pr[p_{\mathrm{true}}\in\mathcal P]\ge1-\alpha\), where the
probability is over the sampled data.
Whenever \(\mathcal P\) is nonempty, the lower confidence bound
\begin{equation}
 \underline\Lambda=\min_{p\in\mathcal P}\Acup(p)
 \label{eq:confidence}
\end{equation}
is the optimum of a linear program. Rejecting the hypothesis
\(\Acup(p_{\mathrm{true}})\le2^k\) only when \(\mathcal P\ne\varnothing\) and
\(\underline\Lambda>2^k\) has false-rejection probability at most
\(\alpha\).
\end{corollary}
\begin{proof}
Minimizing jointly over \(p,q,t\) in \eqref{eq:lp}, with the
additional linear constraints \(p\in\mathcal P\), gives the stated LP.
It is feasible when \(\mathcal P\ne\varnothing\) and bounded below by one,
so its minimum is attained. On the coverage event,
\(\underline\Lambda\le\Acup(p_{\mathrm{true}})\).
Under the hypothesis, rejection therefore implies
\(p_{\mathrm{true}}\notin\mathcal P\), an event of probability at most \(\alpha\).
\end{proof}
For a noiseless message of at most \(k\) bits, the lower bound exceeds
\(2^k\) exactly when no distribution in \(\mathcal P\) admits a
compatible explanation with unrestricted nonsignalling assistance and a freely
chosen internal message, by Theorem~\ref{thm:communication}(ii).
This is an interpretation of the chosen confidence region, not an
optimality claim about its statistical power.
For independent identically distributed trials, binomial intervals
with guaranteed coverage for each conditional probability give such
a region when their noncoverage bounds sum to at most \(\alpha\). The union
bound supplies simultaneous coverage, and intersecting the intervals
with normalization gives \(\mathcal P\). Unobserved pairs \((x,y)\)
receive the interval \([0,1]\), while an empty region is inconclusive.
Replacing \(2^k\) by \(\Lambda(\mathcal N)\) gives the noisy-channel certificate.
For correlated or drifting trials, the confidence regions must instead
be valid for that sampling model.

\section{Proof of the characterization for fixed quantum resources}
\label{app:exactfrontier}
\subsection{Joint distributions and detection}
\label{app:frontier-proof}
Write \(s=(s_1,\ldots,s_n)\in\prod_y\mathcal B_y\) for a list of
receiver outcomes, one for each setting. Bell locality has the following
joint-distribution characterization \cite{PR_Fine1982}.

\begin{lemma}[Fine's joint-distribution characterization of Bell locality
\cite{PR_Fine1982}]
\label{lem:localcoupling}
A nonsignalling correlation \(P(a,b\mid x,y)\) with finite alphabets
is Bell-local if and only if there are conditional distributions
\(Q(a,s\mid x)\) satisfying
\begin{equation}
\begin{aligned}
 \sum_{s:\,s_y=b}Q(a,s\mid x)&=P(a,b\mid x,y),\\
 \sum_a Q(a,s\mid x)&=Q_B(s),
\end{aligned}
\label{eq:localcoupling}
\end{equation}
with \(Q_B\) independent of \(x\).
\end{lemma}
\begin{proof}[Proof for completeness]
A local model gives
\[
 Q(a,s\mid x)=\sum_\lambda p(\lambda)
 p_A(a\mid x,\lambda)\prod_y p_B(s_y\mid y,\lambda).
\]
Conversely, use \(s\) drawn from \(Q_B\) as an input-independent hidden
variable, let the sender use \(Q(a,s\mid x)/Q_B(s)\), and let the
receiver report \(s_y\). The sender's response can be arbitrary where
\(Q_B(s)=0\).
\end{proof}

\begin{recalled}{Theorem}{thm:operationalfrontier}{Characterization for fixed quantum resources}
Under Assumption~\ref{ass:communication}, fix a shared state
\(\rho_{AB}\) and the receiver's measurement families
\(M^{(m)}=(M^{(m)}_{b|y})_{b,y}\), one for each message value
\(m\in\mathcal M\), with \(\abs{\mathcal M}\ge2\).
The finite sender-input and reported-outcome alphabets, the sender's
measurements and the receiver's classical output processing may be chosen.
The following are equivalent:
\begin{enumerate}
 \item A communication experiment using the fixed state and receiver
       measurements can produce probabilities \(p\) with
       \(\Acup(p)>\abs{\mathcal M}\), certifying incompatibility of
       at least one measurement family \(M^{(m)}\).
 \item For some message value \(m\), the receiver's measurements
       \(M^{(m)}\) and a finite family of two-outcome sender POVMs
       \(F_{a|x}\) produce Bell-nonlocal correlations
       \[
       P(a,b\mid x,y)=\tr[(F_{a|x}\otimes M^{(m)}_{b|y})\rho_{AB}],
       \qquad a\in\{0,1\}.
       \]
\end{enumerate}
Whenever these conditions hold, a one-bit protocol achieves
\(\Acup(p)>2\) with the same state, using the receiver's measurements
for the message value in condition~2 and for one other message value.
\end{recalled}

\begin{proof}[Proof of Theorem~\ref{thm:operationalfrontier}]
Write \(\rho_B=\tr_A\rho_{AB}\). For sender POVMs \((A_{m|x})_m\),
the joint probabilities of the message and receiver outcome are
\[
 r(m,b\mid x,y)=\tr[(A_{m|x}\otimes M^{(m)}_{b|y})\rho_{AB}].
\]
Any allowed receiver report is a classical processing of these outcomes.
Independent shared randomness added to the fixed resource does not
change the existence of a violation: by convexity of \(\Acup\), a
violating mixture has a violating component.
It therefore suffices to consider protocols without this additional
randomization.

\emph{\(1\Rightarrow2\).}
Suppose condition~2 fails. For any sender POVMs \((A_{m|x})_m\) and
each message \(m\), consider the binary POVMs
\(\{A_{m|x},\id_A-A_{m|x}\}\), assigning outcome zero to
\(A_{m|x}\). Their correlations with \(M^{(m)}\) are Bell-local
by the failure of condition~2.
Lemma~\ref{lem:localcoupling} gives distributions
\(Q^{(m)}(a,s\mid x)\) with normalized receiver marginals
\(Q_B^{(m)}(s)\), independent of \(x\). Set
\[
 q_m(s\mid x)=Q^{(m)}(0,s\mid x).
\]
The coordinate marginals of \(q_m\) are \(r(m,b\mid x,y)\), and
\(q_m(s\mid x)\le Q_B^{(m)}(s)\). Hence its column maxima sum to
at most one. Define a joint distribution on reported outcome strings by
\[
 q\bigl((m,s_1),\ldots,(m,s_n)\mid x\bigr)=q_m(s\mid x),
\]
and give zero weight to strings with different message labels.
It has marginals \(r(m,b\mid x,y)\) and is normalized because
\(\sum_{m,b}r(m,b\mid x,y)=1\). Hence
\[
 \Acup(r)\le\Lambda(q)
 =\sum_{m,s}\max_x q_m(s\mid x)
 \le\sum_{m,s}Q_B^{(m)}(s)=\abs{\mathcal M}.
\]
Classical processing of the receiver's reported outcome cannot increase
\(\Acup\), by Proposition~\ref{prop:basic}(ii). Thus the compatibility
bound cannot be violated.

\emph{\(2\Rightarrow1\).}
Relabel the message values as \(1,\ldots,\abs{\mathcal M}\), with \(M^{(1)}\)
witnessing condition~2. Let the binary sender POVMs \((F_{a|x})_{a=0,1}\)
produce the Bell-nonlocal correlation
\[
 P(a,b\mid x,y)=\tr[(F_{a|x}\otimes M^{(1)}_{b|y})\rho_{AB}],
 \qquad a\in\{0,1\}.
\]
On these original inputs, perform this POVM and send message \(1\) for
outcome zero and message \(2\) for outcome one, and assign zero
probability to the other messages. For each \(m\), add a new input
\(x_m\) on which the sender sends \(m\) deterministically. Keep the
prescribed receiver measurements and report \((m,b)\). In particular,
\(M^{(2)}\) need not equal \(M^{(1)}\). The resulting probabilities obey
\begin{equation}
\begin{aligned}
 r(1,b\mid x,y)&=P(0,b\mid x,y),\\
 r(m',b\mid x_m,y)&=\delta_{m'm}\tr(\rho_B M^{(m)}_{b|y}).
\end{aligned}
\label{eq:frontierconstruction}
\end{equation}

Suppose \(\Acup(r)\le\abs{\mathcal M}\), and choose a minimizing joint distribution
\(q\in\mathcal C(r)\). Its outcome lists contain reported pairs.
Write \(q_m(s\mid x)\) for its weight on the list
\(((m,s_1),\ldots,(m,s_n))\). These weights need not sum to one over
\((m,s)\), since \(q\) may also assign weight to lists with different
message labels at different receiver settings. At input \(x_m\), the
zero marginals of all other messages force \(q\) to be supported on
the constant-message lists labelled \(m\), so
\(\sum_s q_m(s\mid x_m)=1\). Their disjoint supports imply
\[
 \abs{\mathcal M}\ge\Lambda(q)
 \ge\sum_{m,s}\max_x q_m(s\mid x)
 \ge\sum_{m,s}q_m(s\mid x_m)=\abs{\mathcal M},
\]
where the maxima include every original and additional input.
All omitted terms and all differences
\(\max_x q_m(s\mid x)-q_m(s\mid x_m)\) are nonnegative.
Equality forces them to vanish. Thus \(q\) assigns no weight to lists
with different message labels at different settings, and
\(q_m(s\mid x)\le q_m(s\mid x_m)\) for all \(m,s,x\).

For each original input define
\[
 Q(0,s\mid x)=q_1(s\mid x),\qquad
 Q(1,s\mid x)=q_1(s\mid x_1)-q_1(s\mid x).
\]
These terms are nonnegative and sum to the normalized,
\(x\)-independent distribution \(q_1(s\mid x_1)\).
Their outcome-zero marginals are \(P(0,b\mid x,y)\), and their
outcome-one marginals are
\[
 \tr(\rho_B M^{(1)}_{b|y})-P(0,b\mid x,y)=P(1,b\mid x,y),
\]
by nonsignalling. Lemma~\ref{lem:localcoupling} would make \(P\)
Bell-local, a contradiction. Therefore \(\Acup(r)>\abs{\mathcal M}\).

Finally, repeat
the construction using only message values \(1,2\) and the two additional
inputs \(x_1,x_2\). The same argument gives \(\Acup(r)>2\) on the
same state, using only the two prescribed families of receiver measurements.
\end{proof}

The argument extends to a fixed nonsignalling resource when the allowed
local measurements include deterministic outputs and are closed under
classical processing of outcomes, including adding zero-probability
labels. Additional input labels only select measurements already allowed
on that resource. Under these conditions the construction above uses
only nonsignalling, positivity of probabilities and these classical
operations, so the same equivalence holds without the quantum description.

\section{Proofs for private randomness}
\label{app:privacy}

Here $B$ denotes the classical receiver outcome, with values $b$.
The receiver's complete quantum system is denoted by $R$ to distinguish it
from this outcome. All entropies and logarithms are in bits.

\subsection{Private output entropy}
\label{app:privacy-entropy}

\begin{recalled}{Theorem}{thm:privacy}{Private output entropy}
In the quantum scenario of Section~\ref{sec:privacy}, let
$p=p(m,b\mid x,y)$ be the joint probabilities of the actual message and
receiver outcome. Put
$\delta=[\Acup(p)-\abs{\mathcal M}]_+$.
Then
\[
 H(B\mid EXYM)\ge
 1-h_2\!\left(\frac12+\frac{2\delta}{Nn(n-1)}\right).
\]
Here $H$ is conditional von Neumann entropy and $h_2$ is binary entropy.
In particular, $\Acup(p)>\abs{\mathcal M}$ certifies positive conditional
output entropy.
\end{recalled}

\begin{proof}[Proof of Theorem~\ref{thm:privacy}]
Write
$T(\rho,\sigma)=\frac12\norm{\rho-\sigma}_1$ for trace distance.
The receiver's state, including the received message, is
$\tau_x=\bigoplus_m\sigma_{m\mid x}$, where
$\tr\sigma_{m\mid x}=p(m\mid x)$ and
$\sum_m\sigma_{m\mid x}=\rho_R$ is independent of $x$.
Absorb receiver preprocessing into the measurements and use a common
projective dilation, including its fixed pure ancilla in $\tau_x$.
The projectors $P_{b\mid y}$ preserve each message subspace and leave the
adversary's conditional states unchanged. Define
\[
 \Delta_y(\tau)=\sum_bP_{b\mid y}\tau P_{b\mid y},\qquad
 d=\frac1{Nn}\sum_{x,y}T\bigl(\tau_x,\Delta_y(\tau_x)\bigr).
\]
Thus $d$ is the average change in the receiver's state when a measurement
is performed and its outcome discarded. We first bound the violation in terms of $d$, then relate $d$
to the conditional entropy.

Perform all receiver measurements in the order $1,\ldots,n$ or its
reverse, with equal probability, retaining the message and every
outcome. This defines a parent measurement for each message. Its
coordinate marginals $c(m,b\mid x,y)$ therefore satisfy
$\Acup(c)\le\abs{\mathcal M}$ by Theorem~\ref{thm:communication}(i).
Write
\[
 \varepsilon=\frac1{2Nn}\sum_{x,y,m,b}
       \abs{p(m,b\mid x,y)-c(m,b\mid x,y)}
\]
for their average total-variation distance. The triangle inequality
and trace-distance contractivity bound the disturbance before any
measurement by the sum of the individual disturbances of its
predecessors. In the two reverse orders, each other measurement
precedes it exactly once. Averaging over orders and inputs gives
$\varepsilon\le(n-1)d/2$.

Choose a minimizing joint distribution $q\in\mathcal C(c)$, so that
$\Lambda(q)\le\abs{\mathcal M}$. Each coordinate is a reported pair
$(m,b)$. For each input, adjust one coordinate marginal at a time from
$c$ to $p$ by a maximal coupling. This changes that coordinate with
probability equal to the total-variation distance of its old and new
marginals, while leaving all other coordinate marginals unchanged.
Moving mass $t$ within one input row increases the sum of column maxima
by at most $t$: only increases in the row entries can raise a column
maximum, and their sum is at most the moved mass. The adjusted joint
distribution therefore gives
\[
 \Acup(p)\le\Acup(c)+Nn\varepsilon
 \le\abs{\mathcal M}+Nn\varepsilon.
\]
Taking the positive part yields
\begin{equation}
 \delta\le Nn\varepsilon\le\frac{Nn(n-1)}2\,d.
 \label{eq:privacy-disturbance}
\end{equation}

\Needspace{8\baselineskip}
For the entropy step, consider a normalized state $\tau$ and a binary
projective measurement, temporarily suppressing $x,y$.
Writing $U=P_0-P_1$ gives
$\Delta(\tau)=(\tau+U\tau U)/2$.
Writing $S$ for von Neumann entropy, the increase
$S(\Delta(\tau))-S(\tau)$ is the Holevo
information of the two equiprobable states $\tau$ and $U\tau U$.
Their trace distance is $2T(\tau,\Delta(\tau))$, so
$T(\tau,\Delta(\tau))\le1/2$ and their minimum discrimination
error is $1/2-T(\tau,\Delta(\tau))$
\cite[Theorem~3.4]{PR_Watrous2018}.
Let $Z\in\{0,1\}$ be the uniformly random state label and
$\widehat Z$ the optimal guess. Then
$H(Z\mid\widehat Z)=H(Z\mathbin\oplus\widehat Z\mid\widehat Z)
\le h_2(\Pr[Z\ne\widehat Z])$.
Data processing thus bounds the Holevo information from below by
\[
 S(\Delta(\tau))-S(\tau)
 \ge g\bigl(T(\tau,\Delta(\tau))\bigr),\qquad
 g(t):=1-h_2(1/2+t),\quad 0\le t\le1/2.
\]

Giving the adversary a purification can only reduce conditional
entropy. The measurement-entropy identity of
\cite[Theorem~1(i), Eq.~(12)]{PR_Coles2012}, with the public message
included in the adversary's information, therefore gives
\begin{equation}
 H(B\mid EM,X=x,Y=y)
 \ge S(\Delta_y(\tau_x))-S(\tau_x)
 \ge g\bigl(T(\tau_x,\Delta_y(\tau_x))\bigr).
 \label{eq:privacy-coherence}
\end{equation}
Averaging over the uniform inputs and using convexity and monotonicity
of $g$ yields
\[
 H(B\mid EXYM)\ge g(d)
 \ge 1-h_2\!\left(\frac12+
       \frac{2\delta}{Nn(n-1)}\right),
\]
which proves the theorem.
\end{proof}

The message does not contribute private entropy: the adversary holds a
public copy of $M$. Since $\tau_x$ is already block diagonal in $M$, reading the message
register does not change it. Reading $M$ and performing the binary
receiver measurement, with their outcomes discarded, therefore gives
the same state $\Delta_y(\tau_x)$ as the binary measurement alone.
This is why the binary entropy estimate applies although the test
retains the message.

\subsection{Composable extraction}
\label{app:privacy-composability}

\begin{recalled}{Corollary}{cor:privacy-extraction}{Composable private randomness}
For the protocol in Section~\ref{sec:privacy-sequential}, including its communication, isolation and
disclosure conditions, and for all sufficiently large $L$, there is a
fixed output length $\ell$ satisfying
\[
 \begin{aligned}
  \ell&\ge Lr(w_0)-O(\sqrt L),\\
  r(w_0)&=1-h_2\!\left(\frac12+
       \frac{2w_0}{Nn(n-1)}\right).
 \end{aligned}
\]
The constant in $O(\sqrt L)$ depends only on the fixed alphabets,
coefficients, acceptance threshold and $\varepsilon$, not on
the device dimensions.

Let $\sigma^{\mathrm{acc}}_{KE\mathcal T}$ be the subnormalized final
state on acceptance, whose trace is the acceptance probability, and
let $\sigma^{\mathrm{acc}}_{E\mathcal T}
=\tr_K\sigma^{\mathrm{acc}}_{KE\mathcal T}$.
The extracted string satisfies
\begin{equation*}
 \frac12
 \left\|\sigma^{\mathrm{acc}}_{KE\mathcal T}
 -\frac{\id_K}{2^\ell}\otimes\sigma^{\mathrm{acc}}_{E\mathcal T}\right\|_1
 \le\varepsilon.
\end{equation*}
\end{recalled}

\begin{proof}
For generalised entropy accumulation, take $A_i=B_i$ as the private
output, $C_i=(X_i,Y_i,M_i,B_i)$ as the unpublished test record, and $R_i$
as the receiver memory. The side information $E_i$ contains all sender
and adversary systems and preserved classical copies of $X^i,Y^i,M^i$.
Here $X^i=(X_1,\ldots,X_i)$, and similarly for the other sequences.
These copies account for final disclosure without allowing additional
signals during the experiment. In each round, the trusted controller
generates independent uniform inputs, independently of all incoming
systems and any purifying reference.
The maps $\mathcal M_i:R_{i-1}E_{i-1}\to C_iA_iR_iE_i$ satisfy
\[
 \operatorname{Tr}_{C_iA_iR_i}\circ\mathcal M_i
 =\mathcal R_i\circ\operatorname{Tr}_{R_{i-1}}
\]
for channels $\mathcal R_i:E_{i-1}\to E_i$: after discarding receiver
systems, only independent input generation and sender operations remain.
Reading $A_i$ and the records in $E_i$ recovers $C_i$ without disturbance.
Since the records are preserved, $C^L$ is likewise recoverable from
$A^L E_L$. These are the nonsignalling and reconstruction conditions of
\cite[Eq.~(4.2), Theorem~4.3]{PR_Metger2024}.

For $t\le Nn(n-1)/4$, define
\[
 r(t)=1-h_2\!\left(\frac12+
       \frac{2[t]_+}{Nn(n-1)}\right).
\]
This function is convex and nondecreasing. The disturbance bound
\eqref{eq:privacy-disturbance}, with $d\le1/2$, gives
$\delta\le Nn(n-1)/4$. Together with
\eqref{eq:privacy-witness}, this ensures that every quantum-realizable
$W(p)$ lies in the stated domain. Theorem~\ref{thm:privacy} therefore gives
$H(B\mid EXYM)\ge r(W(p))$, since $[W(p)]_+\le\delta$.
The same argument holds for every incoming round state, with any
purifying reference included in $E$, because the inputs are generated
independently of all these systems.

For a distribution $q$ on the round record $C_i$, put
\[
 f(q)=r(w_0)+r'(w_0)
 \left(Nn\sum_{x,y,m,b}q(x,y,m,b)w_{xymb}-w_0\right).
\]
For a realizable $q$, uniform inputs give
$Nn\sum_{x,y,m,b}q(x,y,m,b)w_{xymb}=W(p)$.
Convexity gives $f(q)\le r(W(p))\le H(B\mid EXYM)$ for every realizable
$q$, so $f$ is an affine min-tradeoff function
\cite[Definition~4.2]{PR_Metger2024}. It is defined for all $q$, with
finite slope by the strict upper bound on $w_0$. For the empirical
distribution $q$ of an accepted sequence, $f(q)\ge r(w_0)$.

Let $p_{\mathrm{acc}}=\tr\sigma^{\mathrm{acc}}_{KE\mathcal T}$ and
$\varepsilon_s=\varepsilon/8$. We write $H_{\min}^{\varepsilon_s}$
for conditional min-entropy smoothed in purified distance.
If $p_{\mathrm{acc}}\ge\varepsilon/2$, let $\rho_{\mathrm{acc}}$
denote the normalized state of $B^L E_L$ after the rounds,
conditioned on acceptance.
Theorem~4.3 of \cite{PR_Metger2024}, with R\'enyi parameter
$\alpha=1+L^{-1/2}$, gives
\[
 H_{\min}^{\varepsilon_s}(B^L\mid E_L)_{\rho_{\mathrm{acc}}}
 \ge Lr(w_0)-O(\sqrt L).
\]
The constant depends only on the fixed alphabets, coefficients,
acceptance threshold and $\varepsilon$, since the output is binary and
$f$ has fixed finite range. The acceptance-dependent correction can be bounded using
$p_{\mathrm{acc}}\ge\varepsilon/2$. The final counts, or any function of them, take at most
$(L+1)^{2Nn\abs{\mathcal M}}$ values. The chain rule for smooth
min-entropy under disclosure of classical information
\cite[Proposition~5.10, Eq.~(5.9)]{PR_Tomamichel2012Thesis} therefore
subtracts at most $2Nn\abs{\mathcal M}\log_2(L+1)$ from the bound,
without changing $\varepsilon_s$. Inputs and messages are already
included in $E_L$.

Let $h$ denote the resulting uniform lower bound on smooth min-entropy
after these disclosures. Choose
$\ell=\lfloor h-2\log_2(2/\varepsilon)\rfloor$, which is nonnegative
for sufficiently large $L$ and does not depend on the observed data.
Conditioned on acceptance, two-universal hashing with an independent
public seed sampled after data generation gives trace distance at most
\[
 2\varepsilon_s+\frac12\,2^{(\ell-h)/2}\le\frac{\varepsilon}{2}
\]
from the product of the uniform state on $K$ and the actual marginal
state of the adversary and public transcript
\cite[Theorem~6 and proof of Lemma~4]{PR_Tomamichel2011}.
The term $2\varepsilon_s$ accounts for smoothing the state and its marginal.
Multiplying by $p_{\mathrm{acc}}$ proves \eqref{eq:privacy-security}
for $p_{\mathrm{acc}}\ge\varepsilon/2$. Otherwise, the subnormalized
distance is at most $p_{\mathrm{acc}}<\varepsilon/2$, since both
states have that trace, including when $p_{\mathrm{acc}}=0$.

\end{proof}

\Needspace{24\baselineskip}
\section{Quantum limitation of the nonsignalling converse}
\label{app:quantum-converse-gap}

\subsection{A quantum counterexample}

\begin{proposition}[Quantum limitation of the converse]
\label{prop:quantum-converse-gap}
For \(x,y\in\{1,2,3\}\) and \(b\in\{0,1\}\), define
\[
 p(1\mid x,y)=
 \begin{cases}
  4/5,&x=y,\\
  1/5,&x\ne y,
 \end{cases}
 \qquad p(0\mid x,y)=1-p(1\mid x,y).
\]
Then the following hold.
\begin{enumerate}
\item[(i)] \(\Acup(p)=2\).
\item[(ii)] \(p\) admits a realization with shared entanglement and one
classical bit.
\item[(iii)] In every quantum realization of \(p\) with one classical
bit, the receiver's measurements are incompatible for at least one
message value, regardless of the shared state or its dimension.
\end{enumerate}
\end{proposition}

\begin{proof}[Proof of Proposition~\ref{prop:quantum-converse-gap}]
We first realize \(p\) using shared entanglement and one classical bit.
Choose three unit Bloch vectors \(\mathbf u_x\) in the \(XZ\) plane
with \(\mathbf u_x\cdot\mathbf u_y=-1/2\) for \(x\ne y\), and
share \(\ket{\Phi^+}=(\ket{00}+\ket{11})/\sqrt2\).
The sender measures \(\mathbf u_x\cdot\boldsymbol\sigma\), where
\(\boldsymbol\sigma\) is the vector of Pauli matrices, and sends
the outcome \(m\in\{-1,1\}\). The receiver uses the effects
\[
 M^{(m)}_{1|y}=\frac25\bigl(\id+m\,\mathbf u_y\cdot\boldsymbol\sigma\bigr),
 \qquad M^{(m)}_{0|y}=\id-M^{(m)}_{1|y}.
\]
The first effect has eigenvalues zero and \(4/5\), so these are valid
POVMs. The conditional receiver state has Bloch vector \(m\mathbf u_x\),
giving \(p(1\mid x,y)=\frac25(1+\mathbf u_x\cdot\mathbf u_y)\).

To evaluate \(\Acup(p)\), consider any joint distribution
\(q\in\mathcal C(p)\), with \(s\in\{0,1\}^3\), and set
\(t_s=\max_x q(s\mid x)\). Define
\[
 W(p)=1+\frac23\sum_xp(1\mid x,x)
             -\frac12\sum_{x\ne y}p(1\mid x,y),
 \qquad
 c_x(s)=\frac13+\frac23s_x-\frac12\sum_{y\ne x}s_y.
\]
Normalization and the marginal constraints give
\[
 W(p)=\sum_{x,s}c_x(s)q(s\mid x)
 \le\sum_s t_s\sum_x\max\{c_x(s),0\}
 \le\sum_s t_s.
\]
The first inequality discards negative terms and uses
\(q(s\mid x)\le t_s\). If \(s\) contains zero, one, two or three
ones, the sum of positive coefficients is respectively \(1,1,1,0\),
which proves the second inequality. The stated probabilities give
\(W(p)=2\). Conversely, the joint distribution
\[
 q(000\mid x)=\frac15,\qquad
 q(e_x\mid x)=\frac25,\qquad
 q(e_x+e_y\mid x)=\frac15\quad(y\ne x),
\]
with all other entries zero, has the required marginals and satisfies
\(\Lambda(q)=1/5+3(2/5)+3(1/5)=2\).
Here \(e_x\) is the binary string whose only nonzero coordinate is
\(x\). This proves \(\Acup(p)=2\).

Suppose a compatible quantum realization with one bit existed. Let
\(E_x\) be the sender's effect for message zero. Measuring the receiver's
parent for message zero or one prepares subnormalized sender states
\(T_s\) or \(U_s\), respectively, with
\(\sum_sT_s=\sum_sU_s=\rho_A\). The induced joint distribution obeys
\[
 q(s\mid x)=\tr\!\bigl(E_xT_s+(\id-E_x)U_s\bigr)
 \le r_s:=\tr(T_s+U_s),\qquad \sum_s r_s=2.
\]
All inequalities in
\(2=W(p)\le\Lambda(q)\le\sum_s r_s=2\) are therefore equalities.
In particular, \(\max_xq(s\mid x)=r_s\) for every \(s\).
Equality in the witness bound also forces \(q(s\mid x)=0\) whenever
\(s\ne000\) and \(s_x=0\), because the corresponding coefficient
\(c_x(s)\) is negative. Hence
\[
 q(000\mid x)=1-p(1\mid x,x)=\frac15=r_{000}
 \qquad\text{for every }x.
\]

For any input attaining \(q(s\mid x)=r_s\), positivity gives
\[
 E_xT_s=T_s,\qquad E_xU_s=0.
\]
Indeed, equality makes both nonnegative terms
\(\tr((\id-E_x)T_s)\) and \(\tr(E_xU_s)\) vanish.
For \(s=000\), these identities hold for every input. Choosing a
maximizing input for each other \(s\) therefore gives
\(T_{000}U_s=U_{000}T_s=0\). Summing over \(s\) yields
\(T_{000}\rho_A=U_{000}\rho_A=0\).
Since both \(T_{000}\) and \(U_{000}\) are positive and bounded above
by \(\rho_A\), they must vanish, contradicting \(r_{000}=1/5\).
The argument does not restrict the shared state or its dimension.

\end{proof}

\end{document}